\documentclass[
superscriptaddress,
 amsmath,amssymb,
 pra,
preprint
floatfix
]{revtex4-2}
\usepackage{graphicx,units,bm}
\usepackage{braket}
\usepackage{physics}
\usepackage{silence}
\usepackage[nohyperlinks,nolist]{acronym}
\usepackage[hidelinks]{hyperref}
\usepackage[dvipsnames]{xcolor}
\usepackage{mathrsfs}
\usepackage{setspace} 

\usepackage{mathtools} 
\usepackage{amsthm}
\newtheorem{theorem}{Theorem}
\newtheorem{proposition}{Proposition}
\newtheorem{lemma}{Lemma}
\usepackage{import}
\usepackage{xifthen}
\usepackage{transparent}

\usepackage{xcolor}

\newcounter{counterexample}
\newcommand{\increase}{
    \stepcounter{counterexample}
}
\makeatletter
\newcommand{\comentario}[1]{
    \refstepcounter{counterexample}
    \label{#1}
}
\makeatother
\newcommand{\resetexample}{
    \setcounter{counterexample}{0}
}

\DeclareMathOperator{\Fidelity}{F}
\DeclareMathOperator{\Variance}{V}
\DeclareMathOperator{\gl}{GL}
\DeclareMathOperator{\range}{Ran}
\DeclareMathOperator{\Span}{span}
\DeclareMathOperator{\diag}{diag}
\DeclareMathOperator{\Aut}{Aut}
\newcommand{\projector}{\ensuremath{\mathbb{P}}}
\newcommand{\pairgroup}{\ensuremath{\mathbb{G}}}
\newcommand{\pairirrep}{\ensuremath{\sigma}}
\newcommand{\myi}{\ensuremath{\mathrm{i}}}
\newcommand{\kket}[1]{\vert #1 \rangle\!\rangle}
\newcommand{\bbra}[1]{\langle\!\langle #1\vert}
\NewDocumentCommand{\repX}{g}{
  \IfNoValueTF{#1}
    {\sigma_X}
    {\sigma_X(#1)}
}
\NewDocumentCommand{\repC}{g}{
  \IfNoValueTF{#1}
    {\sigma_{C_9^{\times 2}}}
    {\sigma_{C_9^{\times 2}}(#1)}
}

\begin{document}

\title{Benchmarking of universal qutrit gates}

\author{David Amaro-Alcal\'a}
\email{david.amaroalcala@ucalgary.ca}
\affiliation{Institute for Quantum Science and Technology, University of Calgary,
 Alberta T2N~1N4, Canada}
 \affiliation{Department of Physics, Lakehead University, Thunder Bay, ON, P7B 5E1}
\author{Barry C. Sanders}
\email{sandersb@ucalgary.ca}
\affiliation{Institute for Quantum Science and Technology, University of Calgary,
 Alberta T2N~1N4, Canada}
\author{Hubert de~Guise}%
\email{hubert.deguise@lakeheadu.ca}
 \affiliation{Department of Physics, Lakehead University, Thunder Bay, ON, P7B 5E1}
\affiliation{Institute for Quantum Science and Technology, University of Calgary,
 Alberta T2N~1N4, Canada}

\date{\today}

\begin{abstract}
We introduce a characterisation scheme for a universal qutrit gate set.
Motivated by the rising interest in qutrit systems,
we apply our criteria to establish that our hyperdihedral group underpins a scheme to characterise the performance of a qutrit T~gate.
Our resulting qutrit scheme is feasible, as it requires resources and data analysis techniques similar to resources employed for qutrit Clifford randomised benchmarking.
Combining our T~gate benchmarking procedure for qutrits with known qutrit
Clifford-gate benchmarking enables complete characterisation of a universal
qutrit gate set.
\end{abstract}

\maketitle

\section{Introduction}\label{sec:introduction}

Driven by the desire to exploit every precious dimension of Hilbert space that
Nature provides~\cite{Wang2020},
the study and development of \(d\)-level systems (qudits) as extensions to qubits are rapidly developing.
Traditional quantum information processing primarily centers on encoding,
manipulating, and reading qubits~\cite{MikeAndIke}.
Qudit experiments are now done
using photons~\cite{Imany2019,Lanyon_Weinhold_Langford_OBrien_Resch_Gilchrist_White_2008}, trapped ions~\cite{randall2015,leupold18,Klimov_Guzman_Retamal_Saavedra_2003},
superconducting qutrits~\cite{Roy2022,Kononenko2021, Morvan2021}, 
dopants in silicon~\cite{fernandez2022coherent}, ultracold atoms~\cite{Lindon2022}, 
and spin systems~\cite{Fu2022}.
\increase

\comentario{reb:clarification-category-mistake}
For reliable qutrit technology, gate characterization, akin to qubit gates, is crucial~\cite{MagesanEaswar2012Emoq}.
An accepted standard of gate characterisation is randomised benchmarking.
\Ac{rb} schemes, in general, 
are used by experimentalists to
estimate the mean 
average gate fidelity over a \ac{gs}~\cite{MagesanEaswar2012Emoq}.
To date,
an explicit extension of \acl{rb} has only been reported for the Clifford
\ac{gs}~\cite{Jafarzadeh_Wu_Sanders_Sanders_2020}.
Here, we extend the \acl{rb} scheme for a universal qutrit \ac{gs}.
\increase

Qudit applications include
quantum teleportation~\cite{Luo_Zhong_Erhard_Wang_Peng_Krenn_Jiang_Li_Liu_Lu_2019, Hu_Zhang_Liu_Cai_Ye_Guo_Xing_Huang_Huang_Li_2020}, quantum memories~\cite{Vashukevich_Bashmakova_Golubeva_Golubev_2022,Otten_Kapoor_Ozguler_Holland_Kowalkowski_Alexeev_Lyon_2021}, Bell-state
measurements~\cite{Zhang_Zhang_Hu_Liu_Huang_Li_Guo_2019}, spin
chains~\cite{Senko_Richerme_Smith_Lee_Cohen_Retzker_Monroe_2015,Blok2021,Hu_Zhang_Liu_Cai_Ye_Guo_Xing_Huang_Huang_Li_2020,Imany2019,Zhang_Zhang_Hu_Liu_Huang_Li_Guo_2019,Lanyon_Weinhold_Langford_OBrien_Resch_Gilchrist_White_2008},
and quantum computing~\cite{Chi2022}.
Qutrits offer advantages over qubits such as:
superior security for quantum
communication~\cite{Bechmann-Pasquinucci_Peres_2000} or avoiding Hilbert-space truncation of
a higher-dimensional system~\cite{Wood2018}.
\increase

For quantum computing, a universal gate set is essential to efficiently approximate any gate~\cite{MikeAndIke}.
Adding a specific gate to the Clifford gate set achieves universality~\cite{Delfosse_2017,Howard2014}.
Such gate is the so-called T gate, which is a non-Clifford member of the third
level of the Clifford hierarchy.
\comentario{reb:contextuality}Interestingly,
  contextuality presents another avenue for universal quantum
computation~\cite{Howard2014,Pavicic2023quantum}.
\increase

\comentario{reb:three-is}Our work is of interest to experimental groups working
with a qutrit set of gates.
We introduce a RB scheme to characterise a universal set of qutrit gates thereby
helping determine the scalability~\cite{Knill2008} of a qutrit platform.
Furthermore, quantum information theorists will find interest in our relaxation
of the
unitary 2-design condition in a qutrit RB scheme~\cite{Morvan2021}.
\increase

\comentario{reb:case-against-novelty}
Our approach extends beyond the qubit case's geometrical considerations,
as detailed in prior studies \cite{dugas2015, barends2014}.
In this context,
we articulate the prerequisites for an optimal generalisation of dihedral benchmarking.
Our findings facilitate the qudit generalization of the dihedral scheme by
identifying and broadening the essential features needed by a gate set to
characterize T gates effectively.

Whereas methods for characterising arbitrary
gate sets are available \cite{chen2022,helsen22prx},
our work is particularly significant for two reasons.
Firstly, we introduce the construction of a gate set that demands minimal resources,
specifically necessitating only X and T gates.
Secondly, we establish criteria that enable the identification or construction of a gate set capable of characterising a T gate.
\increase

Qutrit experiments are proliferating~\cite{Imany2019,Lanyon_Weinhold_Langford_OBrien_Resch_Gilchrist_White_2008,randall2015,leupold18,Klimov_Guzman_Retamal_Saavedra_2003,Roy2022,Kononenko2021,Morvan2021,fernandez2022coherent,Lindon2022,Fu2022},
and our extension to experimentally feasible \acl{rb} schemes for characterising qutrit T gates
ushers in full characterisation of universal qutrit gates.
Furthermore, our method sets the stage for extending experimental
characterisation of a universal set of gates to qutrit cases.
With respect to \acl{rb} theory, our results offer a complete characterisation of the generators of a universal qutrit \ac{gs}.
\increase

Our work is preceded by the qubit case, 
wherein characterisation of a T gate is done via \acl{db}~\cite{barends2014,dugas2015}.
Dihedral benchmarking
twirls (i.e., averages over the uniform measure of a group)~\cite{bennet96,MagesanEaswar2012Emoq} over a representation
of the dihedral group $D_8$~\cite{Tinkham92}.
Here, we generalise the \ac{db} scheme to qutrits.
\comentario{reb:number-primitive-gates}Our scheme is optimal with respect to the
number of primitive gates required (X,
T, and H);
\comentario{reb:notation-gates}we use upright letters for gates and slanted font
for the corresponding
matrices.
\increase

\comentario{reb:new-sentence}
We now start the discussion of our extension of \acl{db} to qutrit systems.
We justify our focus on qutrits as currently there is
  an increasingly number of qutrit
  implementations~\cite{Imany2019,Lanyon_Weinhold_Langford_OBrien_Resch_Gilchrist_White_2008,randall2015,leupold18,Klimov_Guzman_Retamal_Saavedra_2003,Roy2022,Kononenko2021,
Morvan2021,fernandez2022coherent,Lindon2022,Fu2022}.
Furthermore, in the qutrit case, we know that our generalisation of $D_8$ (and
the corresponding \ac{irrep}) is the unique pair 
leading to an optimal generalisation of dihedral benchmarking.
Before establishing our qutrit generalisation of the dihedral group, we first recall several key mathematical concepts in \acl{rb} schemes for qutrits.
\increase

\section{Background}\label{sec:background}
\resetexample
We introduce part of the key algebraic entities needed in our work.
We work in the three-dimensional Hilbert space \(\mathscr{H} \coloneqq
\mathrm{span}(\ket{0},\ket{1}, \ket{2})\).
A state \(\rho\) 
is a positive trace-class operator with trace 1~\cite{moretti2017}.
The set of states in \(\mathscr{H}\) is denoted \(\mathscr{D}\);
pure states are the extreme points of \(\mathscr{D}\) and are of the form
\begin{equation}
\rho_\psi \coloneqq  \dyad{\psi}, \quad \ket{\psi}\in \mathscr{H}.
\end{equation}
Then for a mapping \(\mathcal{E}\colon \mathscr{D}\to \mathscr{D}\) and 
\(\ket{\Omega} \coloneqq  \nicefrac{1}{\sqrt{3}}(\ket{00}+\ket{11}+\ket{22})\),
the Choi-Jamio{\l}kowski operator is $\mathscr{J}_{\mathcal{E}}
\coloneqq 
3(\mathcal{E}\otimes \mathbb{I})\left(\dyad{\Omega}\right)$.\\
\increase

\comentario{reb:clarification-reps-groups}
We now describe the representation of the algebraic objects introduced in the
previous paragraph.
We denote by  \(\Lambda_{\mathcal{E}}\) the  matrix representation (in the
computational basis of \(\mathscr{H}^{\otimes 2}\)) of \(\mathscr{J}_{\mathcal{E}}\).
States~\(\{\rho\}\) are represented by a nine-dimensional vector \(\kket{\rho}\)
satisfying \(\kket{\mathcal{E}(\rho)} = \Lambda_{\mathcal{E}}
\kket{\rho}\); \(\kket{\rho}\) 
is computed by stacking the rows of the matrix representation (in the computational
basis) of \(\rho\)~\cite{choi1975,jamiolkowski1972}.
We emphasise gates are physical objects; therefore, it is incorrect to
discuss their representation.
\increase

The qutrit T gate has an important role within quantum computing.
A T~gate correspond to the action of some
unitary matrix $T \in \mathcal{C}_3\setminus\mathcal{C}_2$~\cite{Wang2020,Cui_Wang_2015,Kitaev_1997},
with $\mathcal{C}_l$ the $l$th level of the qutrit Clifford hierarchy.
For convenience, we only consider diagonal $T$ matrices.
Let \(\omega_d \coloneqq  \exp(2\pi \myi/d)\).
For qutrits,
the generalised Hadamard and  a~T matrix are~\cite{Watson_Campbell_Anwar_Browne_2015}
\begin{align}
H &\coloneqq 
\nicefrac{1}{\sqrt{3}}
\begin{bmatrix}
    1 & 1 & 1 \\
    1 & \omega_{3} & \omega_{3}^2 \\
    1 & \omega_{3}^2 & \omega_{3}
\end{bmatrix},\,
T \coloneqq \mqty[\dmat{1,\omega_{9}^{8},\omega_{9}}],
\label{eq:Tgate}
\end{align}
respectively.
The gates generated by~H and~T,
denoted by the generating set
$\langle \mathrm{H},\mathrm{T}\rangle$,
is a universal
\ac{gs} \cite{Cui_Wang_2015, Kitaev_1997, brylinski2002universal}.
The corresponding set of matrices is denoted by $\langle H, T\rangle$.
\increase

The qutrit Pauli group is defined in terms of the \ac{hw} matrices,
themselves one natural unitary generalisation of the Pauli matrices~\cite{Patera_Zassenhaus_1988}.
The qutrit \ac{hw} matrices are 
powers of the clock and shift matrices~\cite{sylvester,schwinger}:
\begin{equation}\label{eq:def-x-z}
Z \ket{i}\coloneqq\omega_{3}^i \ket{i},
X \ket{i} \coloneqq \ket{i\oplus 1},
i\in [3] \coloneqq
\{0,1,2\},
\end{equation} 
with
\(\oplus\) denoting addition modulo 3 and \([k] \coloneqq  \{0,\ldots, k-1\}\).
In turn, the \ac{hw} matrices
$\{W_k, k\in[9]\}$ correspond to
\(W_{3i+j} \coloneqq X^i Z^j, \text{for }  i,j\in[3]\):
\(\mathcal{W} \coloneqq  \{W_{3i+j}\colon i,j\in [3]\}\).
Then the qutrit Pauli group 
is
\(\mathcal{P} \coloneqq \langle \mathcal{W}, \omega_3 \mathbb{I}\rangle\).
\increase

Several concepts from representation theory are used in our work~\cite{serre1977}.
Given a finite group \(\mathbb{G}\) and a vector space \(\mathscr{V}\), a representation 
\(\sigma\) is a homomorphic mapping
from \(\mathbb{G}\) to \(\gl(\mathscr{V})\); 
henceforth \(\mathscr{V}\) refer to either \(\mathscr{H}\) or \(\mathscr{H}^{\otimes 2}\).
For concreteness, we employ the canonical isomorphisms
\(\gl(\mathscr{H}) \cong \mathscr{M}\) and \(\gl(\mathscr{H}^{\otimes 2}) \cong \mathscr{M}^{\otimes 2}\),
to ensure our representations are matrices.
The range (or image) of \(\sigma\) is denoted \(\range(\sigma) \coloneqq
\{\sigma(g)\colon g\in \mathbb{G}\}\).
\increase

The term irreducible representation can refer to a subspace and to a mapping.
Given a non-trivial subspace \(\Sigma\subseteq \mathscr{V}\) invariant 
under the action of \(\sigma\),
we decompose 
$\mathscr{V} = \Sigma \oplus
\Sigma^{\perp}$,\ where \({}^{\perp}\) denotes orthogonal complement.
In general, 
if \(\sigma\) has an ordered multiset of non-trivial invariant subspaces
\(\{\Sigma_i\}\), 
\(\mathscr{V}\) can be decomposed as
\begin{equation}
\mathscr{V} = \bigoplus_i \Sigma_i.
\end{equation}
The subspaces \{\(\Sigma_i\)\} are known as irreps, mostly in the context of the
decomposition of a representation in irreps~\cite{serre1977}.
Unless specified,
capital Greek letters represent irreps as subspaces,
whereas lowercase Greek letters indicate their homomorphic mappings.
\increase

We now introduce the representation of a group computed from the Choi matrix.
Let \(\mathbb{G}\) be a finite group with a unitary representation \(\sigma\colon
\mathbb{G}\to \mathscr{M}\).
We define the representation \(\Lambda_\sigma\colon \mathbb{G}\to
\mathscr{M}^{\otimes 2}\) that maps \(g\in \mathbb{G}\) to 
\(\Lambda_\sigma(g) \coloneqq \sigma(g)\otimes \sigma(g)^*\), where \({}^*\)
denotes complex conjugation.
We refer to \(\Lambda_{\sigma}(g)\) as a gate.
We sometimes shorten \(\Lambda_{\sigma}(g)\) by \(\Lambda_g\) when the knowledge
of \(\sigma\) is implicit or unnecessary; we follow the convention of using a Greek
subindex to denote the representation and a Latin subindex an element of such
representation. 
\increase

We recall the definition of the twirl by a representation of a group.
Let \(\mathbb{G}\) be a finite group with a  three-dimensional representation
\(\sigma\).
The twirl of a channel
$\mathcal{E}$
over a group \(\mathbb{G}\) is 
\begin{equation}\label{eq:twirl_def}
\mathcal{T}_{\mathcal{E}}^{(\mathbb{G}, \sigma)}
\coloneqq 
\underset{g\in \mathbb{G}}{\mathbb{E}}
\Lambda_g^{\dagger}
\Lambda_{\mathcal{E}}
\Lambda_g,
\end{equation}
where 
\(\mathbb{E}_{x\in \mathbb{X}}\) denotes average over
the uniform measure on \(\mathbb{X}\); that is, \(x\) has probability \(1/|\mathbb{X}|\).
We generally omit the pair group-irrep \((\mathbb{G}, \sigma)\) 
in writing the left-hand side of Eq.~\eqref{eq:twirl_def};
the trace of \(\mathcal{T}_{\mathcal{E}}\) is used in RB schemes to estimate the \ac{agf}. 
\increase

Before passing to the next section,
we define the ideal and noisy versions of a channel labelled by a group element.
Let $g\in \mathbb{G}$, we call $\Lambda_g$ the ideal channel corresponding to \(g\).
Then if \(\mathcal{E}_g\) is a channel associated with the noise accompanying
the action of  the
gate \(\Lambda_g\),
the noisy version of \(\Lambda_g\) is 
\begin{equation} \label{eq:definition_noise}
  \tilde{\Lambda}_g \coloneqq  \Lambda_{\mathcal{E}_g} \Lambda_g.
\end{equation}
Using the tools of representation theory and quantum channels presented above,
we then formulate our generalisation of \ac{db}.
\increase

\section{Approach}
\label{sec:approach}
We are now ready to describe our approach to articulating and solving the problem of benchmarking a universal set of qutrit gates.
First we introduce the HDG as a generalisation of the dihedral group,
needed for generalising qubits to qutrits.
Then we elaborate on our benchmarking scheme for the HDG.
We discuss the formal properties our scheme generalises from the qubit case.
\subsection{Hyperdihedral group}\label{sec:hdg-construction}
\resetexample
We now introduce our generalisation of $D_8$, which we call \ac{hdg}.
Our extension of \ac{db} is based on a unitary irreducible representation (unirrep) of the \ac{hdg}.
We establish this representation in the following two paragraphs.
The \ac{hdg} is the semidirect product, we formally specify the product later, between \(C_3\) and \(C_9^{\times 2}\).
We justify the choice of the \ac{hdg} in Appendix~\S\ref{subsec:criteria}.
We discuss the characterisation of other diagonal gates at the end of this subsection.
\increase

The unirrep for \ac{hdg} is defined using two auxiliary representations.
The first auxiliary representation is \(\repX\colon C_3\to
\mathscr{M}\).
If the abstract elements of the order-three cyclic group \(C_3\) are \( \{a^{k}\colon k\in [3]\} \),
the mapping \(\repX\) is \(\repX{a^{k}} = X^{k}\), where \(X\) is given in
Eq.~\eqref{eq:def-x-z}.
\increase

The second auxiliary representation 
is now introduced and used to define the unirrep we use for the HDG.
Consider the mapping \(\repC\colon
C_9^{\times 2}\to \mathscr{M}\).
If the elements of \(C_9^{\times 2}\) are \(\bm{\alpha} = (\alpha_0, \alpha_1)\in [9]^{\times 2}\),
then \(\repC{\bm{\alpha}} = T^{\alpha_0} (T')^{\alpha_1}\),
where \(T' \coloneqq \diag[\omega_9^{2},\omega_9^{6},\omega_9]\).
Using $\repX$ and $\repC$, 
the \ac{hdg} irrep our scheme uses is
\begin{equation}\label{eq:the-irrep}
\gamma 
\colon \mathrm{HDG} \to \mathscr{M}
\colon
(a^{k},\bm{\alpha}) \mapsto \repX(a^{k})\repC{\bm{\alpha}}.
\end{equation}
Notice \(\range(\gamma) = \langle T, X\rangle\), which is reminiscent of \(D_8\).
\increase

We now provide the definition of the HDG.
Consider the automorphism \(\phi \in \Aut(C_9^{\times 2})\) is
\begin{equation}
\phi(T) \coloneqq  T^{3}(T')^{4},
\quad 
\phi(T') \coloneqq T^{8}(T')^{5}.
\end{equation}
Considering \(\phi\), the HDG is completely defined by
\begin{equation}
  \mathrm{HDG} \coloneqq
  C_3\ltimes_{\phi}C_9^{\times 2};
\end{equation}
the mapping \(\phi\) depends on \(T\).
Additional details can be found in Appendix~\ref{app:computation-hdg-elemes}.
\increase

\comentario{reb:fewer-than-clifford}
We discuss several properties of the \ac{hdg} and the resulting RB scheme.
The \ac{hdg},
consisting of 243 group elements,
requires only 81 gates when global phases are removed~\cite{Morvan2021}.
Consequently,
our scheme uses fewer gates than Clifford-based RB schemes.
It is worth mentioning that \(H \not\in \range(\gamma)\),
which is a property shared with \ac{db}.

Our scheme has another two additional characteristics useful in practical
settings.
The entire set of HDG gates is generated solely by the X and T gates,
which also enjoy a simplified multiplication rule between group elements.
The \ac{agf} and \ac{sp},
derived from averaging over the \ac{hdg},
are dependent on two complex parameters.
\increase

The \ac{hdg} is a natural generalization of \(D_8\);
like \(D_8\),
it has a semidirect product structure~\cite{altmann77}.
As a result of the semidirect product structure of the HDG,
group elements and their products can be straightforwardly expressed as powers of the generating elements, as done in 
Appendix~\ref{app:computation-hdg-elemes}.
Thus, 
sampling from the HDG is straightforward and does not require
approximate methods, as is generally necessary for arbitrary finite groups~\cite{Franca2018}.
\increase

We now discuss the prerequisites of our scheme.
Our scheme requires three primitive gates (X, T, and H), \ac{spam} of 
\(\ket0\) and $\ket+\coloneqq H \ket0$, and the construction of circuits with a depth of up to 200~\ac{hdg} gates.
\comentario{reb:gates-needed}
Among these gates, the X and T gates are the generators of the benchmarked gate
set, whereas the H gate is only required to prepare the state
\(\ket{+}\).
\increase

Current qutrit experiments satisfy the requirements of our scheme~\cite{Kononenko2021,Morvan2021}.
For instance, the \ac{bi}~\cite{Morvan2021} uses primitive gates for
rotations in the subspaces \(\Span(\ket0, \ket{1})\) and \(\Span(\ket{1}, \ket{2})\). 
These authors have also reported the composition of more than 200 qutrit gates.
These characteristics support the claim that our scheme is currently feasible.
\increase

\comentario{reb:benchmarking-other-gates}
Our scheme is not limited to the characterisation of \(T\) in Eq.~\eqref{eq:Tgate}.
By substituting \(T\) by any other diagonal matrix (in the computational basis) with order at least
three, the construction of the HDG can be applied to such gate.
The resulting representation has the same irrep decomposition as the HDG.
Thus, our scheme is useful to characterise any diagonal gate with order at
least three.
\increase

\comentario{reb:why-t-gate}
  We chose to employ the T gate, as defined in Eq.~\eqref{eq:Tgate}, for
  it enables universal quantum computing. Using non-Clifford gates
  like \(T\) is beneficial due to the availability of established
  magic-state distillation procedures for generating such a gate. Furthermore,
  the use of magic-state distillation is notably advantageous, as this method
  has been integrated into error-correcting codes~\cite{Campbell2014}.
\increase

\section{Results}\label{sec:results}
\resetexample
We now provide the expressions for the AGF and the SP resulting from using a HDG
gate set.
These expressions correspond to our genersalisation to qutrits of dihedral benchmarking.
We also show our scheme is made, as Clifford RB schemes are, SPAM-error
independent by adding a projector to the SP expression.
\subsection{Survival probability and average gate fidelity}\label{sec:gateindependent}
\resetexample

\comentario{reb:feasible-agf}
We introduce our scheme to characterise a universal gate set,
  which is our generalisation for qutrits of \ac{db}.
Our scheme feasibly estimates the \ac{agf} of the \ac{hdg} gate set.
Our analysis assumes every gate-set member has the same noise,
which is referred to as gate-independent analysis.
It is worth mentioning that our scheme is compatible with the Fourier~transform
method~\cite{Merkel2021, Wallman_2018}.
We introduce our scheme first by presenting the twirl computed over the
\ac{hdg} and then the expressions for the \ac{agf} and the \ac{sp}.
\increase

We now write the explicit expression of the twirl.
We start by considering the projectors onto the different representation
spaces~\cite{serre1977} in Eq.~\eqref{eq:decomposition-irrep-hdg} of
Appendix~\ref{app:proofs}:
\(
\Pi_{\mathbb{I}},
\Pi_{\Gamma_0},
\Pi_{\Gamma_0^{*}},
\Pi_{\Gamma_+}, 
\Pi_{\Gamma_+^{*}}.
\)
The eigenvalues are
\begin{equation}
  \label{eqs:twirl_entries}
  \lambda_\Gamma(\mathcal{E}) \coloneqq 
  \frac{\trace(\Lambda_{\mathcal{E}}\Pi_\Gamma)}{\trace(\Pi_\Gamma
  \Pi_\Gamma^{\top})}.
\end{equation}
Then the twirl of a channel $\mathcal{E}$ over the \ac{hdg} is
\begin{equation}
  \mathcal{T}_{\mathcal{E}} =
    \sum_{\Gamma \in \{
\Gamma_{\mathbb I}, 
\Gamma_{0}, 
\Gamma_{0}^*, 
\Gamma_{+}, 
\Gamma_{+}^*
\}}
    \lambda_\Gamma(\mathcal{E}) \Pi_\Gamma.
    \label{eq:matrix_plrep_qutrit}
  \end{equation}
From Eq.~\eqref{eq:matrix_plrep_qutrit}, there are only 
two non-trivial complex entries: $\lambda_0$ and $\lambda_+$. 
Let \(\varsigma\in \{0,+\}\),
then
the parameters $\lambda_\varsigma$ are conveniently written in polar form:
\begin{equation}\label{eq:twirl_entries_explicit}
  \lambda_\varsigma = r_\varsigma \exp(\myi \varphi_\varsigma).
\end{equation}
\increase

We now write the \ac{sp} in our scheme.
The gate-independent conditions means that for all group members
\(g\in\mathrm{HDG}\), the noisy channel has the form
\begin{equation}
  \tilde{\Lambda}_g = \Lambda_{\mathcal{E}}\Lambda_g;
    \label{eq:gate_independent_condition}
\end{equation}
that is, every \ac{gs} member has the same noise channel \(\mathcal{E}\).
Let $\rho$ be a state, $E\in \{\rho, \mathbb{I}-\rho\}$, and \(m\) a positive integer.
Using the assumption of Eq.~\eqref{eq:gate_independent_condition},
the \ac{sp} for the HDG is
\begin{equation}\label{eq:surv_prob_gralgroup}
\Pr(m;\rho, E, \mathrm{HDG}) =
\bbra{E} \Lambda_{\mathcal{E}} \mathcal{T}_{\mathcal{E}}^{m} \kket{\rho}.
\end{equation}
We rewrite Eq.~\eqref{eq:surv_prob_gralgroup} knowing
\(\mathcal{T}_{\mathcal{E}}\) is diagonal to obtain:
\begin{equation}
    \Pr(m;\rho, E, \mathrm{HDG}) =
    \sum_{\Gamma}
    \lambda_\Gamma^{m}(\mathcal{E})
    \bbra{E} \Lambda_{\mathcal{E}} \Pi_\Gamma \kket{\rho},
\label{eq:survivalexpanded}
\end{equation}
where the sum is over the irreps in the decomposition of Eq.~\eqref{eq:decomposition-irrep-hdg}.
\increase

We now show how Eq.~\eqref{eq:survivalexpanded} is 
used to estimate,
from the circuit depth \emph{vs}
\ac{sp} curve, the \ac{agf} over \ac{hdg}.
We obtain the expression for the \ac{sp} curve,
which is a decaying exponential function.
To express the \ac{sp} as a function of the twirl entries,
we consider the states
\begin{equation}\label{eqs:initial_states}
    \rho_\varsigma \coloneqq \ketbra{\varsigma},\quad \varsigma\in \{0,+\}.
\end{equation}
Substituting \(\rho\)
in Eq.~\eqref{eq:survivalexpanded} with $\rho_\varsigma$ given in Eq.~\eqref{eqs:initial_states} we obtain
the \ac{sp}
\begin{equation}\label{eq:survival_probability_gateindep}
       \Pr(m ;\rho_\varsigma, \rho_\varsigma, \mathrm{HDG})
    =
    \frac{1}{3} + \frac{2}{3} b_\varsigma r_\varsigma^m \cos(m \varphi_\varsigma).
\end{equation}
Thus, the \ac{sp} can be used to estimate the \ac{agf}.
\increase

At this point we now 
introduce the \ac{agf}
and
relate it to the \acp{sp} written in 
Eqs.~\eqref{eq:survival_probability_gateindep}.
The \ac{agf} computed over a group \(\mathbb{G}\) is defined as:
\begin{equation}
    \bar \Fidelity \coloneqq
    \underset{g\in G}{\mathbb{E}}\Fidelity(\tilde{\Lambda}_g, \Lambda_g),
    \label{eq:def-agf}
\end{equation}
where $\Fidelity(\tilde{\Lambda}_g, \Lambda_g)$ is the gate fidelity between the ideal 
and the noisy channel corresponding to $\Lambda_\sigma(g)$.
In general, for any pair of qutrit channels $\cal E$ and $\mathcal{E}'$, the 
\ac{agf} $\Fidelity(\mathcal{E}', \cal E)$ is~\cite{NIELSEN2002249}
\begin{equation}
    \Fidelity({\mathcal{E}'}, {\cal E})
    \coloneqq
    \frac{1}{12} \trace{\Lambda_{\mathcal{E}'}^{\dagger} \Lambda_{\mathcal{E}}}
      +
    \frac{1}{4}.
    \label{eq:definition_agf}
\end{equation}
\increase

Next, we write the \ac{agf} in terms of the twirl parameters.
For gate-independent benchmarking, the quantity estimated by \ac{hdg} benchmarking~\cite{dugas2015} is 
the \ac{agf} 
between the twirl and the identity
\begin{equation}
    \bar \Fidelity =  \Fidelity(\mathcal{T}_{\mathcal{E}}, \mathbb{I}) = 
\frac{1}{12} \trace\mathcal{T}_{\mathcal{E}}+\frac{1}{4},
    \label{eq:operation_definition_fidelity}
\end{equation}
where $\mathcal{T}_{\mathcal{E}}$ 
is defined in Eq.~\eqref{eq:twirl_def}.
For the qutrit \ac{hdg}, using Eqs.~\eqref{eq:survival_probability_gateindep}:
\begin{equation}
    \bar \Fidelity = \Fidelity(\mathcal{T}, \mathbb{I}) \\= 
    \frac{1}{12}(1 + 2 r_0\cos\varphi_0+ 6 r_{+}\cos\varphi_{+}) + \frac{1}{4}.
    \label{eq:rbnumber_qutrit}
\end{equation}
Note how the quantities \(b_\varsigma\) in 
Eq.~\eqref{eq:survival_probability_gateindep} are not needed to estimate the \ac{agf}.
\increase

It is possible to neglect the phases in Eq.~\eqref{eq:rbnumber_qutrit}:
we justify in Appendix~\ref{app:real_part}
that for  high-fidelity configurations $\varphi_0\ll 1$ and $\varphi_{+}\ll 1$.
Thus, we simplify Eq.~(\ref{eq:rbnumber_qutrit}) to
\begin{equation}
    \bar \Fidelity = \Fidelity(\mathcal{T}, \mathbb{I}) \approx
    \frac{1}{12}(1 + 2 r_0+ 6 r_{+}) + \frac{1}{4}.
    \label{eq:rbnumber_qutrit_real}
\end{equation}
Notice that the previous approximation for the \ac{agf} is always valid.
However, for large values of \(m\), 
the single exponential approximation for the \ac{sp} could fail; 
we study the validity of the single-exponential approximation in Appendix~\ref{app:real_part}.
\increase

\subsection{Removal of SPAM-error contributions}
\resetexample
An important feature of Clifford \acl{rb} schemes is their independence of \ac{spam} errors~\cite{gambetta2012}.
However, the \ac{hdg} \ac{sp} given in Eq.~\eqref{eq:survival_probability_gateindep} is not \ac{sef}.
One way to overcome this limitation is by computing a projector~\cite{Franca2018} that,
when multiplied with the twirl, leads to an expression of the survival probability with a single parameter;
thus our scheme is \ac{sef}.
\increase

The projector-based method for term removal is not the only option and may sometimes be superfluous.
There are known alternatives to this approach~\cite{Helsen_Xue_Vandersypen_Wehner_2019, claes2021}.
Furthermore,
if a \ac{gs} achieves a fidelity of approximately \( \bar{\Fidelity} \approx 0.99 \),
the need for SPAM removal techniques diminishes,
as illustrated in \S\ref{section:numerics} and explored in other studies~\cite{chen2022}.
\increase

We now compute the projectors.
The Choi matrix of the X and Z gates satisfies
\begin{equation}\label{eqs:projectors_spamfree}
\projector_\varsigma \coloneqq 
\sum_{k\in [3]} \Lambda_{Q_{\varsigma}^{k}},
\end{equation}
where $\varsigma\in\{0,+\}$, $Q_0\coloneqq X$,  $Q_+ \coloneqq Z$.
The projectors \(\mathbb{P}_{\varsigma}\) satisfy:
\(\mathbb{P}_+ \Pi_{\Gamma_0} = \mathbb{P}_+ \Pi_{\Gamma_0^{*}} = \mathbf{0}\)
and \(\mathbb{P}_0 \Pi_{\Gamma_+} =  \mathbb{P}_0\Pi_{\Gamma_+^{*}} = \mathbf{0}\), where
\(\mathbf{0}\) is the null matrix in \(\mathscr{M}^{\otimes 2}\).
\increase

By multiplying 
\(\mathcal{T}_{\mathcal{E}}\)
from the left by
\(\projector_\varsigma\)
we remove every parameter in Eq.~\eqref{eq:matrix_plrep_qutrit} except
 \(\lambda_\varsigma\) .
Therefore, using a modified \ac{sp}---with powers of the clock and shift matrices---%
we can obtain a \ac{sp}
which depends
only on selected parameters, independently 
of the initial state and the final measurement.
We compute such \ac{sp} in the next paragraph.
\increase

The modified \ac{sp}---by which we mean including 
the projectors in
Eq.~\eqref{eqs:projectors_spamfree}---is
\begin{align}
   &\check{\Pr}(m; \rho, E, \mathrm{HDG};\varsigma,k)
   \coloneqq
   \bbra{E}\Lambda_{Q^{k}_{\varsigma}} \Lambda_{\mathcal{E}}
   \mathcal{T}_{\mathcal{E}}^{m} \kket{\rho},
        \label{eq:p_check}
\end{align}
where $k\in[3]$.
Using Eq.~\eqref{eq:p_check},
we reach the \ac{sef} version of the
\acp{sp}~\eqref{eq:survival_probability_gateindep}:
\begin{equation}\label{eqs:survivalprobability_gatedependent_SPAMfree}
{\Pr}_{\varsigma}^{\rm SEF} \coloneqq
\sum_{k\in [3]} \check{\Pr}(m; \rho, E, \mathrm{HDG};\varsigma,k)
=
\lambda_{\mathbb I} + 2\Re{(\lambda_\varsigma^m \alpha_\varsigma^{\text{SEF}})},
\end{equation}
where \(\alpha_\varsigma^{\mathrm{SEF}}\in \mathbb{C}\) are constants absorbing \ac{spam} contributions.
Eq.~\eqref{eqs:survivalprobability_gatedependent_SPAMfree} shows that, even if the coefficients
depend on the initial state preparation, the eigenvalues \(\lambda_\varsigma\) remain unchanged so that the expression
of $\bar \Fidelity$ in Eq.~\eqref{eq:def-agf} is \ac{sef}.
\increase

\section{Numerics}\label{section:numerics}
\resetexample
In this section we numerically investigate the feasibility of our scheme.
Our study is done by comparing the variance of the Clifford and HDG gate sets.
This is done using experimental resources reported for a transmon
qutrit~\cite{Morvan2021}.
Our results show that both variances are qualitatively similar.
Thus, given that the experimental resources required for our scheme are similar to those of Clifford RB,
if Clifford RB can be implemented,
our scheme can likewise be appropriately executed.

\subsection{Noise model}\label{subsec:error_model}
\resetexample
We introduce examples of channels used to add noise to \ac{hdg} gates.
These channels are motivated by the features of the BI and the noise models
presented elsewhere~\cite{Wallman_2018}.
\increase

In determining the appropriate noise for each gate,
we observe the following distinction:
elements within the \ac{hdg} fall into two distinct categories—diagonal matrices and powers of $X$.
Notably,
the $X$ gates' implementation differs from that of the diagonal gates~\cite{Blok2021}.
Given this difference,
we introduce specific noise types for each: for the diagonal matrices,
we incorporate noise by adding a phase to the state \(\ket{1}\),
whereas for the powers of $X$, we introduce an over-rotation error.
\increase

In our example, the over-rotation error corresponds to adding a phase to the state $\ket{1}$.
Thus we
represent this noise by conjugating a state by the following unitary matrix:
\begin{equation}\label{eq:error-model-phif}
  U_{\varphi}
  \coloneqq 
\begin{bmatrix}
  1 & & \\
  & \exp(\myi \varphi) & \\
  & & 1
\end{bmatrix}.
\end{equation}
For the over-rotation error,
the mapping corresponds to the conjugation of an state by a matrix of the form
\begin{equation}\label{eq:error-model-psiF}
  U_{(\psi, \xi)}
  \coloneqq 
    V \exp(-\myi \psi M_{01}) \exp(-\myi \xi M_{12}) V^{\dagger},
\end{equation}
where \(\psi,\xi\in [0,2\pi)\),
$V\sim \mathsf{Haar}({\rm SU}(3))$ is a matrix
randomly sampled using the ${\rm SU}(3)$ Haar measure~\cite{deguise2018}, and
\begin{equation}
    M_{01} \coloneqq \begin{bmatrix}
    1 & & \\
      &-1 & \\
      & & 1
    \end{bmatrix},\quad
    M_{12} \coloneqq \begin{bmatrix}
    1 & & \\
      &1 & \\
      & & -1
    \end{bmatrix}.
\end{equation}
\increase

\subsection{Survival probability statistics}
\resetexample
We now analyse numerically the \ac{sp}. 
The primary objective of this examination is to highlight the similarities between the variances of the Clifford and \ac{hdg} \acp{gs}. 
Notably,
similar variance behaviours suggest a similar number of samples required across both methods. 
Given our reliance on a small subset of the Clifford gate set,
excluding \(T\),
the feasibility of our scheme is related to this sample count. 
\increase

Through a numerical analysis,
we determine the variance of the \ac{hdg} \ac{sp} when subjected to noise.
This is important, as the determination of the number of samples
required depends on the variance.
Although there is a model for the
variance~\cite{Helsen2019}, it includes numerous
parameters; these parameters limit its practicality.
\increase

Alternatively, 
a two-parameter empirical model is stated
for qubit Clifford \ac{rb}~\cite{Itoko_Raymond_2021}.
Unfortunately, 
this latter model asymptotically approaches zero.
This behaviour is not seen in qutrits,
where the variance converges to a non-zero value.
\increase

The variance of the \ac{sp} is
\begin{equation}
\Variance(m;\rho, E, \mathbb{G})
\coloneqq
\bbra{E^{\otimes 2}}
\Lambda_{\mathcal{E}}^{\otimes 2}
(\mathcal{T}_{\mathcal{E}^{\otimes 2}})^{m} \kket{\rho^{\otimes 2}}
\\-
\bbra{E^{\otimes 2}}
\Lambda_{\mathcal{E}}^{\otimes 2}
(\mathcal{T}^{\otimes 2}_{\mathcal{E}})^{m}
\kket{\rho^{\otimes 2}},
\end{equation}
where 
$m$ is the circuit depth,
\(\mathcal{T}_{\mathcal{E}^{\otimes 2}} \coloneqq \mathbb{E}_g \Lambda_g^{\otimes 2} \Lambda_{\mathcal{E}}^{\otimes 2}(\Lambda_g^{\otimes 2})^{\dagger} \),
\(\rho\) is the initial state, and $E$ the final measurement~\cite{Helsen2019}.
An example of the variance of the \ac{sp} (for qutrit Clifford and HDG)
is presented in Fig.~\ref{fig:random_variable}.
\increase

\begin{figure}[ht]
    \centering
    \includegraphics{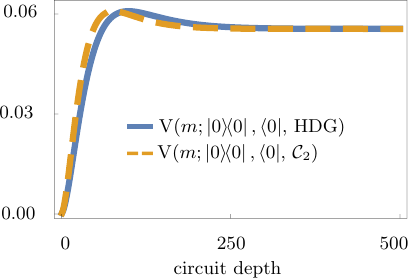}
    \caption{
    This figure shows the variance for two \acp{sp}, differentiated by 
    the \ac{gs} used in the simulation: blue and orange lines
    correspond to the \ac{hdg} and Clifford \acp{gs}, respectively;
    the variable \(m\) in the argument of \(\mathrm{V}\) denotes the circuit
    depth.
    We observe that 
    the two curves are qualitatively similar.
     In the plot we present the variance for the SP
     for a configuration with $\Fidelity(\Lambda_{U_{\varphi}}) = 0.9999$ and
     $\Fidelity(\Lambda_{(\psi,\psi)}) = 0.99$ (in Eqs.~\eqref{eq:error-model-phif} and \eqref{eq:error-model-psiF}), using $\ket{0}$ as initial state.
    }
    \label{fig:random_variable}
\end{figure}

These numerical results show that the variance of the \ac{sp} 
for the \ac{hdg} \ac{gs} is qualitatively similar to the Clifford case.
Consequently, it is reasonable to expect that the number of required samples
should be comparable.
\increase

\section{Discussion}\label{sec:discussion}
\resetexample
Our extension of the \acl{rb} scheme to characterise T~gates
is given by the expressions ~\eqref{eq:survival_probability_gateindep}
and \eqref{eqs:survivalprobability_gatedependent_SPAMfree}, together with the
qutrit \ac{hdg} \acp{gs}.
To summarise our steps, we obtained the expression for the \ac{hdg} \ac{agf} in
Eq.~\eqref{eq:rbnumber_qutrit_real}.
We then showed that the parameters of the \ac{hdg} \ac{agf} are accessible 
via fit from the survival probabilities in
Eqs.~\eqref{eq:survival_probability_gateindep}
and Eqs.~\eqref{eqs:survivalprobability_gatedependent_SPAMfree} respectively
for ideal 
and noisy---subject to \ac{spam} errors---initial states.
\increase

Next we examined the experimental resources required for our scheme.
Compared with the 216 gates of the Clifford group, the 81 gates of the quotient
\ac{hdg}/\(\langle \omega_3\mathbb{I}\rangle\) 
reduce by $\sim2/3$ the number of gates required for benchmarking and
provide
a more efficient scheme than \acl{ib},
with respect to the gates needed to be synthesised~\cite{MagesanEaswar2012Emoq}.
\increase

We then analysed the practical properties of our scheme.
By enforcing the condition of a diagonal twirl, we simplified the data analysis required
for computing the \ac{agf}.
This is especially clear compared to the non-diagonal cases~\cite{Franca2018,helsen22prx}. 
Additionally, the semidirect product structure of the HDG allows the efficient sampling of \ac{hdg} elements, eliminating the need for approximate Markov chain methods~\cite{Franca2018}.
Finally, we asserted our scheme is feasible as it is based on gates (X, H, and T) that can be implemented by current 
platforms~\cite{Morvan2021,Kononenko2021}.
\increase

In \S\ref{section:numerics},
we simulated our scheme using the experimental parameters from a transmon qutrit~\cite{Morvan2021}.
Our findings indicate that the statistics of the \ac{hdg} \ac{sp} closely
resemble those of the Clifford \ac{gs}~\cite{Helsen2019}.
Consequently,
comparable experimental resources—such as measurements and the number of randomly sampled circuits—are required,
and the same statistical tools can be employed.
\increase

\comentario{reb:interleaved-benchmarking}
 We conclude our discussion with a comment on 
non-Clifford interleaved
benchmarking~\cite{dugas2015,MagesanEaswar2012Emoq}.
The HDG can be used to characterise diagonal gates.
However, our schemes and the construction of the HDG, cannot be used to
characterise the X gate.
The reason is that, by removing the X gate from the HDG, we obtain an abelian
subgroup. Twirling by an abelian group leads to a twirl with more parameters than for the
HDG~\cite{amaro2024}.
\increase

\section{Conclusions}
\resetexample
We have extended the \acl{rb} scheme to characterise qutrit T gates.
Our scheme relies on our generalisation of the dihedral~group for qubits, which
we call the hyperdihedral~group.
Using the hyperdihedral group, we derived 
 closed-form expressions
for the \acl{sp} and \acl{agf} for gate sets that include a qutrit T gate.
Our scheme characterises a diagonal qutrit T gate, the non-Clifford generator of a
universal qutrit \ac{gs}.
Thus, our extension completes the characterisation of a universal qutrit \ac{gs}.
Finally, to prove our scheme's feasibility, we simulated its application on a
transmon qutrit T gate~\cite{Morvan2021}.

DAA, BCS, and HdG acknowledge support from 
Natural Sciences and Engineering Research Council of Canada
and the Government of Alberta.

\bibliography{manuscript.bib}
\appendix
\section{Effect of phases on the survival probability and average gate fidelity}\label{app:real_part}
In this section, we study the effect of the phases in
Eq.~\eqref{eq:twirl_entries_explicit} on the survival probability curve and
average gate fidelity.
We show that, for high-fidelity gates, the contribution of the phases can be neglected in the AGF.
However, for high-depth circuits, the survival probability curve deviates from a single exponential.

To consider the most general case, we express the phase in terms of the
\(\chi\)-representation.
This representation has its origins in quantum
tomography~\cite{MikeAndIke,chuangchirep}.
We use this representation to obtain the general expression for the phase in eigenvalue \(\lambda_0\) in Eq.~\eqref{eq:twirl_entries_explicit}.
In the \(\chi\)-representation, the phase in Eq.~\eqref{eq:twirl_entries_explicit} is given by
\begin{equation}
  \varphi_0 = 
2\arctan\left(\frac{v}{\sqrt{u^2 + v^2}+u}\right), 
\end{equation}
where  $u\coloneqq \Re(\lambda_0)$, and $v\coloneqq \Im(\lambda_0)$.
Specifically, we have
\begin{subequations}\label{eqs:realandimaginarypart}
\begin{align}
  u &= \chi_{00} + \chi_{11} + \chi_{22} \\&- \frac{1}{2}(\chi_{33} + \chi_{44} + \chi_{55}+\chi_{66} + \chi_{77} + \chi_{88}), \nonumber\\
  v &= \frac{2}{\sqrt{3}}\left( \chi_{33} + \chi_{44} + \chi_{55} -\chi_{66} -\chi_{77} - \chi_{88} \right).\label{eq:imaginarypart}
\end{align}
\end{subequations}

We now analyse the asymptotic behaviour of the phase for high-fidelity gates.
High-fidelity implies $\chi_{00} \lessapprox 1$ and $\chi_{ii} \ll 1$ for $i>0$. This implies that $u\approx 1$ and $v \ll 1$.
Asymptotic behaviour of \(\cos\varphi_0\) is
\begin{equation}\label{eq:app-cos-varphi}
  \cos\varphi_0 = 1-\frac{1}{2}\left(\frac{v}{u}\right)^2 + \mathcal{O}((u/v)^{4})\\ \approx 1 - \frac{1}{2}\left(\frac{1-\bar\Fidelity}{\bar\Fidelity}\right)^2.
\end{equation}
Eq.~\eqref{eq:app-cos-varphi} shows that for high-fidelity gates approximation Eq.~\eqref{eq:rbnumber_qutrit_real} is valid.

However, the experimental estimate of the eigenvalues \(\lambda_\varsigma\)
could be affected by the phase in Eq.~\eqref{eq:twirl_entries_explicit}.
We study the case when the phase \(\varphi_0\) is maximal for a given AGF.
We consider \(\bar{\Fidelity} = 0.9925\).
In Fig.~\ref{fig:phase} we show the logarithm of the SP.
From Fig.~\ref{fig:phase} we notice a deviation from a straight line; this deviation indicates the SP is not a single-exponential.
The shape of this curve makes it difficult to estimate the parameter if the number of composed gates is large.
\begin{figure}[ht]
  \includegraphics{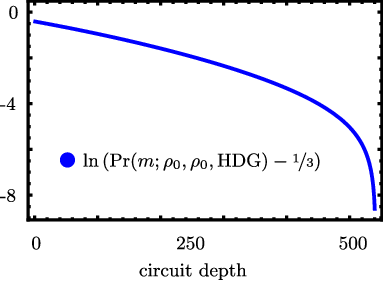}
    \caption{Log plot of the \ac{sp} \(\Pr\)
      of Eq.~\eqref{eq:survival_probability_gateindep} for a gate-independent 
    \ac{rb} run.
    The noise was fixed with \(\bar{\Fidelity} = 0.9925\).
    We highlight the presence of the oscillatory contribution as a deviation from a straight line.
}
\label{fig:phase}
\end{figure}

The contribution of the phase to the SP in Eq.~\eqref{eq:survival_probability_gateindep}
is assessed by an asymptotic expansion.
For large values of the fidelity 
\begin{equation}
\cos m\varphi_0 = 1-\frac{1}{2}(mu/v)^2+\mathcal{O}((u/v)^{4})\\
\approx 1 - \frac{1}{2}(m(1-\bar\Fidelity)/\bar\Fidelity)^2.
\end{equation}
Therefore, even for high-fidelity gates, shallower circuit depths must be considered for the single-exponential fit.
Otherwise the presence of the phase produces a bad fit.

\section{Constructing and manipulating qutrit HDG elements}
\label{app:computation-hdg-elemes}

For convenience, we construct elements in the qutrit \ac{hdg} by computing products of the matrices $X$,
\(A_1 \coloneqq \diag[1 , \omega_{9}^8, \omega_{9}]\),
and
\(A_2 \coloneqq \diag[\omega_{9}^2, \omega_{9}^6,  \omega_{9}]\),
where $\omega_9 \coloneqq\exp(2\pi i/9)$.
Each \ac{hdg} member \(X^x A_1^y A_2^z\) is labelled by the \((x,y,z)\), where $x,y\in \mathbb{Z}_3$,  $z\in \mathbb{Z}_9$.
    For two words 
    \((x_1, y_1, z_1)\) and
    \((x_2, y_2, z_2)\)
    labelling \ac{hdg} elements 
\(X^{x_1} A_1^{y_1} A_2^{z_1}\) and
\(X^{x_2} A_1^{y_2} A_2^{z_2}\) respectively,
the resulting group element is
\(
    X^{x_3} A_1^{y_3} A_2^{z_3} =
    X^{x_1} A_1^{y_1} A_2^{z_1}
    X^{x_2} A_1^{y_2} A_2^{z_2}
    \) where 
\((x_3, y_3, z_3)\) is given by 
\begin{equation}
\label{eq:multiplication_rule_hdg}
\begin{bmatrix}
    x_3\\
    y_3\\
    z_3
\end{bmatrix}
\coloneqq
    \begin{bmatrix}
    1 & 0 & 0 \\
    0 & 5 & 8 \\
    0 & 4 & 3
    \end{bmatrix}^{x_2}
    \begin{bmatrix}
        x_1\\
        y_1\\
        z_1
    \end{bmatrix}
    +
\begin{bmatrix}
    x_2\\
    y_2\\
    z_2
\end{bmatrix}.
\end{equation}
Similarly, for an \ac{hdg} element \((x_1,y_1,z_1)\), the inverse 
word \((x_2,y_2,z_2)\) satisfying
\begin{equation}
    (X^{x_1} A_1^{y_1} A_2^{z_1}) (X^{x_2} A_1^{y_2} A_2^{z_2}) = \mathbb{I}_3
\end{equation}
is given by
\begin{equation}
\begin{bmatrix}
    x_2\\
    y_2\\
    z_2
\end{bmatrix}
\coloneqq
    -\begin{bmatrix}
    1 & 0 & 0 \\
    0 & 5 & 8 \\
    0 & 4 & 3
    \end{bmatrix}^{3-x_1}
    \begin{bmatrix}
        x_1\\
        y_1\\
        z_1
    \end{bmatrix}.
\end{equation}
The multiplication rule in Eq.~\eqref{eq:multiplication_rule_hdg} also hints at the semidirect product structure of the group,
where the $X$ gate, acting by conjugation, is an automorphism for the subgroup generated by the matrices $A_1$ and $A_2$.
\section{Criteria for the selection of the HDG and proofs} \label{app:proofs}
\subsection{Criteria}\label{subsec:criteria}
\resetexample

We now explain our reasoning for choosing the HDG.
We identify and examine four properties that a pair must satisfy for optimal
characterisation of a qutrit T gate within a RB scheme.
We then prove that the pair (\ac{hdg}, $\gamma$) is the unique pair that meets
our
four criteria.
Appropriate and optimal pairs (as defined in the next two paragraphs)
generalise the pair group-irrep used in \acl{db}.
\increase

Our four criteria are divided into two categories: 
two criteria distinguish appropriate from inappropriate pairs whilst the
remaining criteria identify optimal pairs.
We later show  that our criteria lead to the identification of a unique
appropriate and optimal pair.
\increase

We can now discuss our criteria for identifying an appropriate pair $(\pairgroup, \pairirrep)$.
\comentario{reb:criterion-irreps}
Our first criterion justifies why only irreducible, and not reducible, representations are used in RB schemes.
Neither in Clifford or \(D_8\) RB schemes, 
this point is addressed. 
The motivation, as discussed in Appendix~\ref{app:proofs},
is to prevent increasing the number of parameters in the \ac{sp} and \ac{agf}.
A pair $(\pairgroup, \pairirrep)$ satisfies our first criterion (C1) if \pairirrep{} is an irrep and
\(T\in \range(\pairirrep{})\).
\increase

The criterion C1 is motivated by the number of parameters in the AGF and SP. 
We can count the number of parameters using the orthogonality of characters~\cite{Tinkham92}.
In Appendix~\ref{app:proofs} we show that if a reducible representation is used,
the number of parameters is unnecessarily increased.
\increase


Our second criterion (C2)
is established so as to only require 
projector or character techniques to recover SPAM error
independence~\cite{dugas2015,claes2021}.
A pair $(\pairgroup, \pairirrep)$ satisfies C2 if it satisfies C1 and twirling
any channel by \(\Lambda_\pairirrep{}\) yields a diagonal matrix (in the
computational basis).
If a pair $(\pairgroup, \pairirrep)$ satisfies C2 we label it as appropriate.
\increase

\comentario{reb:importance-c-2}
We stress the significance of C2
in light of some alternatives~\cite{helsen22prx,chen2022}.
Whereas RB can be realised with non-diagonal twirls,
our criterion intentionally circumvents the necessity for additional statistical techniques.
This enables our method to be incorporated as a subroutine in a more comprehensive characterisation scheme.
\increase

Our next two criteria deal with experimental costs
\comentario{reb:importance-c-3and4}
and are necessary to pick the best candidate among the appropriate groups identified
with C2.
We introduce our third criterion (C3) to reduce the number of gates needed.
A pair $(\pairgroup, \pairirrep)$ satisfies C3 if it satisfies C2 and the order
of \pairgroup{} is minimal: not other appropriate pair contains a group with fewer
elements than \(\mathbb{G}\).
\increase

For our fourth criterion (C4), we consider \ac{spam} costs. 
A pair $(\pairgroup, \pairirrep)$ satisfies C4 if $(\pairgroup, \pairirrep)$
satisfies C3 and twirling by \pairirrep{} yields a matrix with a minimal number of distinct eigenvalues.
If a pair $(\pairgroup, \pairirrep)$ satisfies C4, we label it as optimal.
\increase

\begin{proposition}
  The following holds for the qutrit pair (\ac{hdg}, $\gamma$):
\begin{itemize}
    \item[P1.] $\gamma$ in Eq.~\eqref{eq:the-irrep} is an \ac{irrep} and $T\in\range(\gamma)$.
    \item[P2.] twirling a channel with respect to $\Lambda_\gamma$ yields a diagonal
      channel in the computational basis;
    \item[P3.] the \ac{hdg} is the group with the smallest order satisfying with
      an \ac{irrep} satisfying P1 and P2;
    \item[P4.] \ac{hdg} \ac{agf} has the smallest number of parameters and satisfies P3.
\end{itemize}
\end{proposition}
\increase

P1 is established through the examination of character properties~\cite{Tinkham92}.
As the sum of the squared moduli of traces of each member of $\range(\gamma)$
equals the order of the \ac{hdg},
$\gamma$ is indeed verified to be an \ac{irrep}~\cite{Tinkham92}.
P2 is confirmed by employing HDG character table to ascertain that 
the irreps of \(\Lambda_\gamma\) decompose  \(\mathscr{H}^{\otimes 2}\)
 as
\begin{equation} \label{eq:decomposition-irrep-hdg}
  \mathscr{H}^{\otimes 2}
=
\Gamma_{\mathbb I}\oplus 
\Gamma_{0}\oplus 
\Gamma_{0}^*\oplus 
\Gamma_{+}\oplus 
\Gamma_{+}^*:
\end{equation} 
 the trivial \ac{irrep} denoted $\Gamma_{\mathbb I}$;
 two conjugated one-dimensional \acp{irrep},~$\Gamma_0$ and $\Gamma_{0}^*$;
 and two conjugated three-dimensional \acp{irrep}, $\Gamma_+$ and $\Gamma_{+}^*$.
Consequently, Schur's lemma (as explicitly analysed in the \ac{sm} of Ref.~\cite{gambetta2012}) ensures the twirl is diagonal.
\increase

We finish the study of Proposition~1 by proving P3 and P4.
P3 is proven by direct enumeration of each group with order smaller than the HDG.
Then since the qutrit \ac{hdg} is the sole group that fulfils P3, P4 follows directly.
As P1 implies C1, P2 implies C2, and P3 and P4 imply C3 and C4,
respectively, Proposition~1 shows the pair (HDG, $\gamma$) satisfies our four criteria 
and is thus a generalisation of $D_8$.
In what follows, we use (\ac{hdg}, $\gamma$) to generalise the \acl{db} scheme.
\increase
\begin{lemma}[\cite{serre1977}]\label{lemma}
Let \(\mathbb{G}\) be a finite group and $\gamma'$ be an irrep  of
\(\mathbb{G}\). Let us define the representation $\Lambda_{\gamma'}\colon \mathbb{G}\to
\mathscr{M}^2\coloneqq g \mapsto 
\gamma'(g)\otimes \gamma'(g)^{*}$. 
Then the trivial irrep (\(g\mapsto 1\in \mathbb{C}\)) appears in the decomposition of $\Lambda_{\gamma'}$.
\end{lemma}

\begin{proof}
Let $\chi(g)$ be the character of the irrep generated by matrices $g\in\langle X,T\rangle$.
Then the character of $\Lambda_{\gamma'}(g)$ is
$\chi_{\Lambda_{\gamma'}}(g) = \chi(g)\chi(g)^* = |\chi(g)|^2$. 
We compute the inner product between $\chi$ and the character of the trivial
representation $\forall g$, $\chi_{\mathbb{I}}(g) = 1$:
\begin{equation*}
  \langle \chi_{\mathbb{I}}, \chi_{\Lambda_{\gamma'}}\rangle  = 
    \frac{1}{|\mathbb{G}|}\sum_g |\chi(g)|^2. 
\end{equation*}
Because $\chi_{\Lambda_{\gamma'}}(1) = d^2$, $\langle \chi_{\mathbb{I}}, \chi_{\Lambda_{\gamma'}}\rangle > 0$.
Therefore, the trivial irrep appears at least once in the decomposition of
$\Lambda_{\gamma'}$~\cite{serre1977}.
\end{proof}

\begin{theorem}
If $\Lambda_{\gamma'}$ is reducible then $\Lambda_{\gamma'}$ does not necessarily produces a diagonal twirl.
\end{theorem}

\begin{proof}
Proving this theorem is equivalent to showing that there is an irrep with multiplicity greater than one in the decomposition of the representation.
Without loss of generality, assume $\Lambda_{\gamma'}$ decomposes into two irreps $\alpha$ and $\beta$ as $\Lambda_{\gamma'} = \alpha\oplus\beta$. Then $\Lambda_{\gamma'} = (\alpha\oplus\beta)\otimes(\alpha\oplus\beta)^* = \alpha\otimes\alpha^* \oplus \alpha\otimes\beta^* \oplus \beta\otimes\alpha^* \oplus\beta\otimes\beta^*$. 
By Lemma~\ref{lemma},
we know that each of the representations, \(\alpha\otimes\alpha^*\) and \(\beta\otimes\beta^*\), carries the trivial irrep at least once.
Thus, $\Lambda_{\gamma'}$ has an irrep with multiplicity at least 2. 
Therefore, $\Lambda_{\gamma'}$ does not necessarily produces a diagonal twirl.
\end{proof}

\begin{acronym}
    \acro{rb}[RB]{randomised benchmarking}
    \acro{db}[DB]{dihedral benchmarking}
    \acro{ib}[IB]{interleaved benchmarking}
    \acro{sp}[SP]{survival probability}
    \acro{hdg}[HDG]{hyperdihedral group}
    \acro{hdb}[HDB]{hyperdihedral benchmarking}
    \acro{hw}[HW]{Heisenberg-Weyl}
    \acro{gs}[gate set]{set of gates}
    \acro{spam}[SPAM]{state preparation and measurement}
    \acro{agf}[AGF]{average gate fidelity}
    \acro{bi}[BI]{Berkeley implementation}
    \acro{pl}[PL]{Pauli-Liouville}
    \acro{irrep}[irrep]{irreducible representation}
    \acro{ft}[FT]{Fourier transform}
    \acro{sef}[SEF]{SPAM error-free}
    \acro{sm}[SM]{Supplementary Material}
    \acro{cptp}[CPTP]{Completely positive trace preserving}
\end{acronym}
\end{document}